\documentclass[aip, jmp, author-year]{revtex4-1}

\usepackage{amsmath}
\usepackage{amsthm}

\usepackage{appendix}
\usepackage{amsmath, amsthm, amssymb}

\usepackage{enumerate}
\usepackage{setspace}
\usepackage{natbib}

\newcommand{\derv}[1]{\frac{\partial}{\partial #1}}
\newcommand{\deriv}[2]{\frac{\partial #1}{\partial #2}}

\newcommand{\beqn}{\begin{equation}}
\newcommand{\eeqn}{\end{equation}}
\newcommand{\beqnar}{\begin{eqnarray}}
\newcommand{\eeqnar}{\end{eqnarray}}

\newtheorem{theorem}{Theorem}[section]

\newtheorem{proposition}[theorem]{Proposition}

\newenvironment{example}[1][Example]{\begin{trivlist}
\item[\hskip \labelsep {\bfseries #1}]}{\end{trivlist}}

\renewcommand{\theequation}{\arabic{section}.\arabic{equation}}
\newcommand{\contr}{\,\lrcorner\,}

\begin{document}

\preprint{Submitted to Journal of Mathematical Physics}

\title{Multi-Symplectic,  Lagrangian, One-Dimensional Gas Dynamics} 



\author{G. M. Webb}
\email[]{gmw0002@uah.edu}
\affiliation{$^1$CSPAR, The University of Alabama in Huntsville,\\ 
Huntsville AL 35805, USA}


\date{\today}

\begin{abstract}
The equations of Lagrangian, ideal, one-dimensional (1D), compressible 
gas dynamics are written in a multi-symplectic form using the  
Lagrangian mass coordinate 
 $m$ and time $t$ as independent variables, and in which the Eulerian 
position of the fluid element $x=x(m,t)$ is one of the dependent variables.
This approach differs from the 
Eulerian, multi-symplectic approach using Clebsch variables.  
Lagrangian constraints are used to specify equations for $x_m$, $x_t$ and $S_t$ consistent with the 
Lagrangian map, where $S$ is the entropy of the gas. 
 We require $S_t=0$ corresponding to advection 
of the entropy $S$ with the flow.  We show that the Lagrangian Hamiltonian equations are related 
to the de Donder-Weyl multi-momentum  formulation. The pullback conservation laws and the symplecticity
conservation laws are discussed. The pullback conservation laws correspond to invariance of the 
action with respect to translations in time (energy conservation) and translations in $m$ in 
Noether's theorem. The conservation law due to $m$-translation invariance 
gives rise to a novel nonlocal conservation law involving the Clebsch variable $r$ 
used to impose $\partial S(m,t)/\partial t=0$.  
 Translation invariance with respect to $x$ 
in Noether's theorem is associated with momentum conservation. We obtain the Cartan-Poincar\'e 
form for the system, and use it to obtain a closed ideal of two-forms representing 
the equation system. 
\end{abstract}



\maketitle 
\section{Introduction}
\setcounter{equation}{0}
In a recent paper, \cite{Webb14c}  obtained a multi-symplectic 
formulation of magnetohydrodynamics (MHD) by using Clebsch variables in an Eulerian variational principle, 
in which the Lagrangian is modified 
by constraints using Lagrange multipliers, that ensure mass, entropy, magnetic flux and the 
Lin constraint are conserved following the flow. The work of \cite{Webb14c} used similar methods 
to \cite{Cotter07} who derived multi-symplectic equations for fluid dynamic type systems.
The work of \cite{Cotter07} uses the Euler-Poincar\'e approach to Hamiltonian fluid type systems 
developed by \cite{Holm98}. \cite{WebbMace15} derive a generalized potential vorticity type 
conservation law in MHD by using Noether's second theorem, in conjunction with a non-field aligned 
fluid relabelling symmetry.  

Work on multi-momentum Hamiltonian systems has a long history, going back to the work of \cite{deDonder35} 
and \cite{Weyl35} where multi-momentum maps and generalized Legendre transformations were introduced 
to generalize Hamiltonian mechanics to a more covariant formulation, where time in a fixed reference frame is 
not the only evolution variable in the equations. The de Donder-Weyl Hamiltonian equations apply to 
action principles in which the Lagrangian $L=L({\bf x},\varphi^i,\partial \varphi^i/\partial x^\mu)$ 
in the independent variables 
${\bf x}$ and the dependent field variables $\varphi^i$ ($1\leq i\leq m$ say),
includes at least two independent derivatives $\partial\varphi^k/\partial x^1$ and $\partial\varphi^k/\partial x^2$
say,  for $\varphi^k$. 
 It turns out that the equations of ideal, 1D Lagrangian gas dynamics can be cast in the 
de Donder-Weyl Hamiltonian form, since the Lagrangian $L$ in the formulation depends on both $x_m$ and 
$x_t$ where $x=x(m,t)$ defines the Lagrangian map, in which $x$ 
is a function of the Lagrangian mass coordinate $m$ and time $t$.  The de Donder-Weyl 
Hamiltonian then depends on the multi-momenta $\pi^t_x=\partial L/\partial x_t$, 
 $\pi^m_x=\partial L/\partial x_m$, $\pi^t_S=\partial L/\partial S_t$ (see e.g. \cite{Kanatchikov97, Kanatchikov98}, 
\cite{Forger03}, \cite{Forger05}). 
 There is an extensive literature on multi-symplectic and multi-momentum Hamiltonian systems
(see e.g. \cite{Gotay91a, Gotay91b}, \cite{Gotay04a, Gotay04b}, \cite{Roman-Roy09}, \cite{Marsden99}, 
\cite{Carinena91}, 
\cite{Bridges05, Bridges10}, \cite{Kanatchikov93, Kanatchikov97, Kanatchikov98}, \cite{Cantrijn99}). 

\cite{Bridges10} relate multi-symplectic structures and generalized Hamiltonian differential 
equation systems  
to the variational bicomplex. The system of equations is written in the generalized 
Hamiltonian form ${\bf X}\contr\omega=-dH$ where ${\bf X}$ is a generalized Hamiltonian vector field, 
$\omega$ is the symplectic form, and $H$ is the generalized Hamiltonian. In this development $H$ can be 
a differential form (i.e. $H$ need not be a scalar) and ${\bf X}$ can be a multi-vector. This approach 
circumvents the need to use a Legendre transformation relating the Lagrangian to the Hamiltonian of the 
system, which is not always obvious for the case of singular Lagrangians. In the variational bicomplex, 
the exterior derivative $d$ is split up into horizontal ($d_h$) and vertical ($d_v$) components (i.e. 
$d=d_h+d_v$). The horizontal 
component of $d_h f$ for a function $f$ involving ${\bf u}$ and their derivatives 
and associated contact one-forms and tensors 
(i.e. the jet space), corresponds to changes on the solution manifold where the ${\bf u}$ explicitly 
depend on the independent variables ${\bf x}$. The vertical component $d_v f$ corresponds to changes in the 
contact one-forms in the jet space which are non-zero off the solution manifold (but are zero on the 
solution manifold). Both $d_h$ and $d_v$ are needed to compute higher order derivatives and  differential 
forms and multi-vectors.  

For the case of singular Lagrangians, Dirac's theory of constraints is useful in determining the Hamiltonian 
and a new Poisson bracket (the Dirac bracket) 
that satisfies the Jacobi identity (e.g. \cite{Chandre13a}, \cite{Chandre13b}).

{\bf One of the benifits of the multi-symplectic formulation of the fluid equations is that 
both space and time can be thought of as evolution variables (e.g. \cite{Bridges05})). 
This is useful in the formulation of travelling wave problems, where two distinct Hamiltonian 
formulations of the equations are possible with, distinct Hamiltonians (e.g. \cite{Bridges92}, 
\cite{Webb05b,Webb07,Webb08,Webb14d}). For systems in one Cartesian space variable $x$
and one time variable $t$, both $x$ and $t$ can be regarded as the evolution variable. 
In the case that time is regarded as the 
evolution variable, the Hamiltonian corresponds to the conserved energy flux, 
in which the conserved $x$-momentum flux acts as a constraint. However, in the case that space 
variable $x$ is the evolution variable, the conserved momentum flux is the Hamiltonian
and the conserved energy flux acts as a constraint. One of the original motivations for the 
multi-symplectic approach to field theory, was the desire for a covariant Hamiltonian formulation
that is frame independent (e.g. \cite{Gotay04a,Gotay04b}). In the usual Hamiltonian approach, 
the space and time variables are first chosen, and the equations are written in Hamiltonian form,
in which time is the evolution variable. This leads to non-covariant evolution equations. 
Multi-symplectic Poisson bracket formulations of multi-symplectic systems have been developed 
for example, by 
(\cite{Kanatchikov93,Kanatchikov97},\cite{Forger03,Forger05}).
} 
 
Our main aim is to obtain multi-symplectic and multi-momentum equations in fluid dynamics 
by working directly with the Lagrangian fluid dynamics  equations 
(see e.g. \cite{Courant76}, \cite{Sjoberg04} for 1D gas dynamics and \cite{Newcomb62} 
for the case of magnetohydrodynamics). \cite{Akhatov91} and \cite{Bluman10} give  
comprehensive accounts of planar gas dynamics, and the 
associated conservation laws of the equations using potential symmetries of the equations and 
related equations (the so-called tree of equations related to the gas dynamic equations) using both the Eulerian 
and Lagrangian form of the equations. \cite{Verosky84} studied higher order symmetries and recursion operators 
for symmetries of the 1D, isentropic gas dynamic equations. \cite{Nutku87} and 
\cite{Olver88} have studied the bi-Hamiltonian and tri-Hamiltonian structure 
of the 1D gas dynamic equations for the case of an isentropic gas ($S=const.$).  

The outline of the paper is as follows. In section 2, we present the equations of 1D Eulerian gas dynamics, including 
a discussion of the first law of thermodynamics. In Section 3, we introduce the Lagrangian map for 1D gas 
dynamics and write the equations in terms of the Lagrangian mass coordinate $m$  and time $t$ as independent variables
where $x=x(m,t)$ gives the Eulerian position of the fluid element. A Lagrangian variational principle  
(e.g \cite{Broer74}, \cite{Webb09}) without constraints is used to describe the system. The 
Euler Lagrange equations for the system give a nonlinear wave equation for $x=x(m,t)$ in which the entropy 
$S=S(m)$ is a pre-specified function of $m$. In Section 4, we develop a constrained variational principle 
using $m$ and $t$ as independent variables, in which the Lagrangian fluid dynamic equations $x_m=\tau$ 
($\tau=1/\rho$), $x_t=u$,  and $S_t=0$ are treated as constraints. To 
ensure that $\partial S(m,t)/\partial t=0$ we add a constraint term $r\partial S/\partial t$ to the 
Lagrangian. Variations of the action with respect to $x$ gives the Lagrangian momentum equation. 
The variational principle requires that $\partial r(m,t)/\partial t=-T$ where $T$ is the 
temperature of the gas. This formulation is used to derive a de Donder-Weyl, multi-momentum Hamiltonian 
description of the system (Section 4.1). In Section 4.2 we derive the multi-symplectic form of the 
equations, which uses the same dependent variables as the de Donder-Weyl formulation. Section 4.3 derives the pullback
and symplecticity conservation laws (e.g. \cite{Hydon05}) for the system. It turns out that one of the pullback 
conservation  laws (which corresponds to translation invariance of the action with respect 
to the Lagrange label $m$) yields a nonlocal conservation law involving the Clebsch variable $r$. 
In Section 4.4 we show how the nonlocal pullback conservation law, and the momentum 
and energy conservation laws arise from Noether's theorem. 
Section 4.5 presents the Cartan-Poincar\'e form for the system. The Cartan-Poincar\'e form can be used 
to obtain a closed ideal of forms representing the partial differential equation system. 
This set of forms can be used to obtain the Lie symmetries and conservation laws of the system using 
Cartan's geometric theory of partial differential equations (e.g. \cite{Harrison71}; 
 \cite{Wahlquist75}). {\bf It remains an open problem, to investigate in detail 
the closed ideal of forms obtained from the Cartan Poincar\'e form. Such 
an investigation will depend on the detailed form of the equation of state for the system 
(i.e. the dependence of the pressure $p$ and energy density $\varepsilon$ on the density $\rho$ and 
entropy $S$ of the system). This will be useful in extending the many investigations 
of 1D gas dynamics for the barotropic gas case, where $p=p(\rho)$, and $\varepsilon=\varepsilon (\rho)$,
to the more general, non-barotopic case.}     Section 5 concludes with a summary and discussion.  
\section{One dimensional gas dynamics}
\setcounter{equation}{0}
The time dependent, ideal, inviscid equations of  Eulerian gas dynamics
in one Cartesian space dimension can be written in the form (\cite{Courant76}, \cite{Sjoberg04}):
\begin{align}
\deriv{\rho}{t}+u\deriv{\rho}{x}+\rho\deriv{u}{x}=&0, \label{eq:2.1}\\
\rho\left(\deriv{u}{t}+u\deriv{u}{x}\right)+\deriv{p}{x}=&0, \label{eq:2.2}\\
\deriv{S}{t}+u\deriv{S}{x}=&0, \label{eq:2.3}
\end{align}
where $u$ is the fluid velocity (assumed to be directed along the $x$-axis), $\rho$ is
the gas density and $p$ is the gas pressure, and $S$ is the gas entropy. 
Equations (\ref{eq:2.1})-(\ref{eq:2.3}) need to be supplemented by an equation of 
state for the gas, (e.g. $p=p(\rho,S)$). The equation of state of the gas is 
related to the first law of thermodynamics:
\begin{equation}
TdS=dQ=dU+pd\tau\quad\hbox{where}\quad \tau=\frac{1}{\rho}. \label{eq:2.4}
\end{equation}
For ideal gas dynamics, $dQ/dt=0$ and $dS/dt=0$ where $d/dt=\partial/\partial t+u\partial/\partial x$ 
is the Lagrangian time derivative moving with the flow. In (\ref{eq:2.4}) $U$ is the internal 
energy of the gas per unit mass, $\tau=1/\rho$ is the specific volume, and $T$ is the 
temperature of the gas.  Using the internal energy per unit volume, $\varepsilon=\rho U$ instead of $U$, the first 
law of thermodynamics reduces to:
\begin{equation}
TdS=\frac{1}{\rho}\left(d\varepsilon-wd\rho\right)\quad\hbox{where}\quad w=\frac{\varepsilon+p}{\rho}, 
\label{eq:2.5}
\end{equation}
is the enthalpy of the gas. Assuming $\varepsilon=\varepsilon(\rho,S)$,  
(\ref{eq:2.5}) gives:
\begin{equation}
\rho T=\varepsilon_S,\quad w=\varepsilon_\rho,\quad p=\rho\varepsilon_\rho-\varepsilon. \label{eq:2.6}
\end{equation}
From (\ref{eq:2.5}): 
\begin{equation}
TdS=dw-\frac{1}{\rho}dp\quad\hbox{and}\quad -\frac{1}{\rho}\nabla p=T\nabla S-\nabla w, \label{eq:2.7}
\end{equation}
which is useful in obtaining the Eulerian energy conservation equation for the system. 
\section{The Lagrangian map}
\setcounter{equation}{0}
In fluid dynamics, the Lagrangian map is defined as the solution of the differential
equation system (e.g \cite{Courant76}, \cite{Broer74}):
\begin{equation}
\frac{d{\bf x}}{dt}={\bf u}({\bf x},t), \label{eq:3.1}
\end{equation}
of the form ${\bf x}={\bf X}({\bf x}_0,t)$ where ${\bf x}={\bf x}_0$ at time $t=0$, in which
the fluid velocity ${\bf u}({\bf x},t)$ is assumed to be a known function of ${\bf x}$ and $t$. 
If ${\bf x}={\bf X}({\bf x}_0,t)$ is a 1-1 invertible map, then the inverse map 
gives ${\bf x}_0={\bf f}({\bf x},t)$ and the fluid velocity ${\bf u}$ 
is given by ${\bf u}=\partial {\bf X}({\bf x}_0,t)/\partial t$, where the time derivative 
is taken with the Lagrange label ${\bf x}_0$ held constant. 

For the case of 1D gas dynamics in one Cartesian space coordinate $x$, the Lagrangian map implies:
\begin{equation}
dx=\deriv{x}{x_0}dx_0+\deriv{x}{t}dt=\deriv{x}{x_0}\left(\deriv{x_0}{x} dx+\deriv{x_0}{t}dt\right)+
\deriv{x}{t}dt. \label{eq:3.2}
\end{equation}
which implies:
\begin{equation}
\deriv{x}{x_0}\deriv{x_0}{x}=1,\quad \deriv{x}{t}+\deriv{x}{x_0}\deriv{x_0}{t}=0. \label{eq:3.3}
\end{equation}
From (\ref{eq:3.3}) we obtain:
\begin{equation}
\deriv{x_0}{t}+u\deriv{x_0}{x}=0, \label{eq:3.4}
\end{equation}
and hence the Lagrange label $x_0$ is advected with the flow. 

The mass continuity equation (\ref{eq:2.1}) may be written in the form:
\begin{equation}
\rho(x,t) dx=\rho_0 dx_0\quad \hbox{or}\quad \rho J=\rho_0, \label{eq:3.5}
\end{equation}
where $\rho_0=\rho(x_0,0)$ and $J=\partial x/\partial x_0$ is the Jacobian of the Lagrangian 
map. The Lagrangian mass coordinate:
\begin{equation}
m=\int_{-\infty}^x \rho(x',t) dx'=\int_{-\infty}^{x_0}\rho(x_0',0)dx'_0, \label{eq:3.6}
\end{equation}
can be used instead of $x_0$ as a Lagrangian label. Note that $m=m(x_0)$ and the entropy $S=S(x_0)$ 
are both advected with the flow. Using the definition of the Lagrangian map and using (\ref{eq:3.6}) 
we obtain the auxiliary relations:
\begin{equation}
x_m=\tau=\frac{1}{\rho}\quad \hbox{and}\quad x_t=u, \label{eq:3.7}
\end{equation}
where we regard $x=x(m,t)$ to be a function of $m$ and $t$.  

A variational principle for 1D gas dynamics is well known 
(e.g. \cite{Courant76}, \cite{Broer74}, \cite{Webb09}). The action 
is defined as:
\begin{equation}
A=\int\int{\cal L}dx\ dt=\int\int{\cal L}_0 dm\ dt, \label{eq:3.8}
\end{equation}
where ${\cal L}$ and ${\cal L}_0$ are defined by:
\begin{equation}
{\cal L}=\frac{1}{2}\rho x_t^2-\varepsilon(\rho,S),\quad {\cal L}_0{\bf (x,m,t)}=x_m {\cal L}=\frac{1}{2} x_t^2-F(x_m,m),\quad 
F(x_m,m)=\frac{\varepsilon(\rho,S)}{\rho}. \label{eq:3.9}
\end{equation}
The Lagrangian $x$-momentum equation for the system is given by the variational equation $\delta A/\delta x=0$, i.e.
\begin{equation}
\frac{\delta A}{\delta x}=-\left(x_{tt}+p_m\right)=0\quad\hbox{where}\quad p=\rho\varepsilon_\rho-\varepsilon=-
\deriv{F(x_m,m)}{x_m}, \label{eq:3.10}
\end{equation}
is the gas pressure. The Eulerian momentum equation $u_t+u u_x=-p_x/\rho$ 
is obtained by multiplying the Lagrangian momentum equation 
(\ref{eq:3.10}) by $\rho=1/x_m$  (note $x_{tt}=u_t+uu_x$).

The Euler Lagrange equation (\ref{eq:3.10}) is in fact a nonlinear wave equation for $x(m,t)$ 
of the form:
\begin{equation}
x_{tt}-a^2\frac{x_{mm}}{x_m^2}+p_S S_m=0, \label{eq:3.11}
\end{equation}
where $a^2=(\partial p/\partial\rho)_S$ is the square of the adiabatic sound speed of the gas. 
For the special case of a gas with a constant adiabatic index $\gamma$, the gas pressure has the form:
\begin{equation}
p=p_1\left(\frac{\rho}{\rho_1}\right)^\gamma \exp\left(\frac{S}{C_v}\right)\quad \hbox{and}\quad \rho=\frac{1}{x_m}. 
\label{eq:3.12}
\end{equation}
In this latter case (\ref{eq:3.11}) reduces to the equation:
\begin{equation}
x_{tt}=N_1x_m^{-\gamma-1}\exp[{\bar S}(m)]\left(\gamma x_{mm}-x_m \bar{S}_m\right). \label{eq:3.13}
\end{equation}
where $N_1=p_1\rho_1^{-\gamma}$ and ${\bar S}=S/C_v$. This is a complicated, nonlinear wave equation for 
$x(m,t)$. However, it reduces to a linear elliptic equation for the case  of the Chaplygin gas 
with $\gamma=-1$ (e.g. \cite{Akhatov91}, \cite{Webb09}). 

Using the first law of thermodynamics (\ref{eq:2.5})-(\ref{eq:2.7}) we obtain the Lagrangian evolution equation 
for the gas pressure as:
\begin{equation}
p_t+\frac{a^2}{\tau^2} u_m=0, \label{eq:3.14}
\end{equation}
where $p=p(m,t)$ is the Lagrangian form for $p$. In the Lagrangian fluid description, the mass continuity 
equation takes the form:
\begin{equation}
u_m=\tau_t. \label{eq:3.15}
\end{equation}
The mass continuity equation (\ref{eq:3.15}) is a consequence of the integrability condition
$x_{mt}=x_{tm}$ in the Lagrangian approach. 

\section{de Donder-Weyl and multi-symplectic approach}
\setcounter{equation}{0}
Consider the action:
\begin{equation}
\hat{A}=\int \hat{L}\ dtdm, \label{eq:dw1}
\end{equation}
where the Lagrangian density ${\bf \hat{L}(x,t,m,u,\tau,S,r)}$ is given by:
\begin{equation}
\hat{L}=L+\lambda(x_m-\tau)+\nu(x_t-u)+rS_t,\quad L=\frac{1}{2}u^2-\varepsilon\tau. \label{eq:dw2}
\end{equation}
The constrained variational principle (\ref{eq:dw1})-(\ref{eq:dw2}) with Lagrange multipliers $\lambda$, $\nu$ and $r$ 
ensures that the basic 1D gas dynamic equations:
\begin{equation}
x_m-\tau=0,\quad x_t-u=0 \quad\hbox{and}\quad S_t=0, \label{eq:dw3}
\end{equation}
are satisfied. Here, $L$ is the unconstrained 
Lagrangian density for 1D gas dynamics, i.e. $L={\cal L}_0$ where ${\cal L}_0$ is given by (\ref{eq:3.9}). 

For the variational principle (\ref{eq:dw1})-(\ref{eq:dw2}) the stationary point conditions for the action 
give the equations:
\begin{align}
\frac{\delta \hat{A}}{\delta\tau}=&\rho\varepsilon_\rho-\varepsilon-\lambda\equiv (p-\lambda)=0, 
\quad \frac{\delta \hat{A}}{\delta\lambda}=(x_m-\tau)=0, \label{eq:dw4}\\
\frac{\delta\hat{A}}{\delta u}=&(u-\nu)=0,\quad \frac{\delta\hat{A}}{\delta\nu}=(x_t-u)=0, \label{eq:dw5}\\
\frac{\delta\hat{A}}{\delta S}=&-\left(r_t+T\right)=0,\quad \frac{\delta\hat{A}}{\delta r}=S_t=0. \label{eq:dw6}
\end{align}
The variational equation:
\begin{equation}
\frac{\delta\hat{A}}{\delta x}=-\left(\nu_t+\lambda_m\right)=-\left(u_t+p_m\right)=0, \label{eq:dw7}
\end{equation}
gives the Lagrangian momentum equation for the system.

Equations (\ref{eq:dw4}) and (\ref{eq:dw5}) give 
\begin{equation}
\lambda=p\quad\hbox{and}\quad \nu=u, \label{eq:dw8}
\end{equation}
for the Lagrange multipliers in (\ref{eq:dw2}). Thus, we can replace $\hat{L}$ in (\ref{eq:dw1}) by:
\begin{align}
\tilde{L}(x,t,m,u,p,r,S)=&\frac{1}{2}u^2-\frac{\varepsilon}{\rho}+p(x_m-\tau)+u(x_t-u)+rS_t\nonumber\\
=&ux_t+p x_m+rS_t-\left(\frac{1}{2}u^2+w\right),
 \label{eq:dw9}
\end{align}
where $w=(\varepsilon+p)/\rho$ is the enthalpy of the gas 
and $\tau=1/\rho$. In (\ref{eq:dw9})  the Lagrange multiplier $r$ ensures that 
the entropy is advected with the flow, i.e. $S_t=0$ where $S=S(m,t)$. 
In our analysis we use the state vector:
\begin{equation}
{\bf z}=(x,u,p,S,r)^T. \label{eq:dw10}
\end{equation}
The internal energy density $\varepsilon=\varepsilon(\rho,S)$, $p=p(\rho,S)$ and 
$w=w(\rho,S)$. However, we prefer to use $p$ and $S$ as independent variables to 
specify the state of the gas. We introduce the function $\tilde{w} (p,S)=w(\rho,S)$
in our analysis. By equating $d\tilde{w}(p,S)$ and $dw(\rho,S)$ we obtain the partial 
differential relations:
\begin{equation}
\tilde{w}_p a^2=w_\rho,\quad \tilde{w}_p p_S+\tilde{w}_S=w_S. 
\label{eq:dw11}
\end{equation}
Using the thermodynamic relations (\ref{eq:2.4})-(\ref{eq:2.7}) we obtain:
\begin{align}
\tilde{w}_p=&\frac{w_\rho}{a^2}=\tau, \nonumber\\
\tilde{w}_S=&w_S-p_S\frac{w_\rho}{a^2}=\frac{\varepsilon_S}{\rho}=T, \label{eq:dw12}
\end{align}
where $T$ is the temperature of the gas. 

Consider the variations of the action $\tilde{A}$  (\ref{eq:dw1}), but with 
$\hat{L}$ replaced by $\tilde{L}$. Using (\ref{eq:dw9}) for $\tilde{L}$ we obtain:
\begin{align}
\frac{\delta{\tilde A}}{\delta x}=&-(u_t+p_m)=0,\quad \frac{\delta{\tilde A}}{\delta u}=(x_t-u)=0, \nonumber\\
\frac{\delta{\tilde A}}{\delta p}=&x_m-\tilde{w}_p=(x_m-\tau)=0, \nonumber\\
\frac{\delta{\tilde A}}{\delta S}=&-\tilde{w}_S-r_t=-(r_t+T)=0,\quad \frac{\delta{\tilde A}}{\delta r}=S_t=0. 
\label{eq:dw13}
\end{align}
Thus, the variational equations (\ref{eq:dw4})-(\ref{eq:dw8}) give the same variational equations as 
(\ref{eq:dw13}) based on the variational derivatives of $\tilde{A}$ and $\tilde{L}$. Both 
approaches  
give the basic equations of 1D Lagrangian gas dynamics, namely:
\begin{equation}
u_t+p_m=0,\quad x_t=u,\quad x_m=\tau,\quad S_t=0,\quad r_t+T=0.  \label{eq:dw14}
\end{equation}
The equation for the Lagrange multiplier $r$ is similar to the Lagrange multiplier 
equation for $\beta/\rho$ in the Eulerian Clebsch variable formulation (e.g. \cite{Zakharov97},
\cite{Morrison98},\cite{Holm83a,Holm83b}, 
\cite{Webb14c}).  
Note that $S_t=0$ following the flow implies $S=S(m)$.  A similar variational principle 
to (\ref{eq:dw9}) was used by \cite{Webb09} in a paper on conservation laws in 1D gas dynamics
(equation (6.20) of that paper). 
\subsection{The de Donder-Weyl formulation}
{\bf A general description of the de-Donder Weyl equations and multi-symplectic geometry 
has been given for example by \cite{Paufler02}, who use the language of fiber bundles. 
We present an elementary derivation of the general form of the de Donder-Weyl Hamiltonian equations 
below. The basic results concern an action principle of the form:
\begin{equation}
J=\int {\cal L}(x^\mu,\varphi^i,\partial_\mu\varphi^i) d{\bf x}, \label{eq:ddw1}
\end{equation}
where the $x^\mu$ are the independent variables, and the $\varphi^i$ are the dependent field variables.
The Euler Lagrange equations for the variational principle (\ref{eq:ddw1}) are:
\begin{equation}
\frac{\delta J}{\delta \varphi^i}=\deriv{\cal L}{\varphi^i} 
-\derv{x^\mu}\left(\deriv{\cal L}{(\partial_\mu\varphi^i)}\right)=0. \label{eq:ddw2}
\end{equation}
The poly-momenta for the system are defined as:
\begin{equation}
\pi^\mu_i=\deriv{\cal L}{(\partial_\mu\varphi^i)}, \label{eq:ddw3}
\end{equation}
In the Hamiltonian description, the multi-symplectic Hamiltonian $H(x^\mu,\varphi^i,\pi^\mu_i)$
is governed by the generalized Legendre transformation: 
\begin{equation}
H=\pi^\mu_i\deriv{{\varphi}^i}{x^\mu}-{\cal L}, \label{eq:ddw4}
\end{equation}
Proceeding as in classical mechanics (e.g. \cite{Goldstein80}), one obtains the balance equations:
\begin{align}
dH=&d\pi^\mu_i \varphi^i_{,\mu}+\pi^\mu_i d\left(\varphi^i_{,\mu}\right)
-\left(\deriv{\cal L}{x^\mu}dx^\mu+\deriv{\cal L}{\varphi^i} d\varphi^i
+\deriv{\cal L}{(\varphi^i_{,\mu})} d\varphi^i_{,\mu}\right)\nonumber\\
=&\deriv{H}{x^\mu} dx^\mu+\deriv{H}{\varphi^i} d\varphi^i+\deriv{H}{\pi^\mu_i}d\pi^\mu_i. 
\label{eq:ddw5}
\end{align}
Equating the coefficients of $dx^\mu$, $d\pi^\mu_i$, $d\varphi^i_{,\mu}$ and $d\varphi^i$ 
gives the balance equations:
\begin{align}
\deriv{H}{x^\mu}+\deriv{\cal L}{x^\mu}=&0, \label{eq:ddw6}\\
\varphi^i_{,\mu}-\deriv{H}{\pi^\mu_i}=&0, \label{eq:ddw7}\\
\pi^\mu_i-\deriv{\cal L}{\varphi^i_{,\mu}}=&0, \label{eq:ddw8}\\
\deriv {\cal L}{\varphi^i}+\deriv{H}{\varphi^i}=&0. \label{eq:ddw9}
\end{align}

Using the Euler Lagrange equations (\ref{eq:ddw2}), the definitions (\ref{eq:ddw3}) 
of the poly-momenta $\pi^\mu_i$ and (\ref{eq:ddw7}) and (\ref{eq:ddw9}) we obtain 
the de Donder-Weyl Hamiltonian equations:
\begin{equation}
\deriv{\varphi^i}{x^\mu}=\deriv {H}{\pi^\mu_i}, 
\quad \derv{x^\mu} \left(\pi^\mu_i\right)=-\deriv{H}{\varphi^i},  
 \label{eq:ddw10}
\end{equation}
which are analogous to Hamilton's equations in ordinary Hamiltonian mechanics, in which the time $t$
is the evolution variable. A more complete description of the de Donder-Weyl equations 
using fiber bundles and jet bundles is given by \cite{Paufler02}.  
 
Below, we use the above approach to multi-symplectic systems,  
 to determine the de Donder-Weyl equations for 1D Lagrangian gas dynamics.} 
We introduce two momenta corresponding to $x(m,t)$ 
and a momentum associated with $S_t$ via the momentum maps:
\begin{equation}
\pi^t_x=\deriv{\tilde{L}}{x_t},\quad \pi^m_x=\deriv{\tilde L}{x_m} \quad\hbox{and}\quad \pi^t_S=\deriv{\tilde L}{S_t}, \label{eq:4.9}
\end{equation}
where the Lagrangian ${\tilde L}$ is given by (\ref{eq:dw9}). We find:
\begin{equation}
\pi^t_x=\deriv{\tilde L}{x_t}=u,\quad \pi^m_x=\deriv{\tilde L}{x_m}=p,\quad \pi^t_S=\deriv{\tilde L}{S_t}=r. \label{eq:4.10}
\end{equation}

Introduce the Hamiltonian $H$ by the generalized Legendre transformation:
\begin{align}
H=&\pi^t_x x_t+\pi^m_x x_m+\pi^t_S S_t-\tilde{L},\nonumber\\
=&u x_t+px_m+rS_t-\left(-\frac{1}{2} u^2+u x_t+px_m+r S_t-w\right)\nonumber\\
=&\frac{1}{2} u^2+w, \label{eq:4.11}
\end{align}
where $w=(\varepsilon+p)/\rho$ is the gas enthalpy. The Hamiltonian (\ref{eq:4.11}) 
is the Eulerian energy flux per unit mass flux, i.e. 
\begin{equation}
\derv{t}\left[\frac{1}{2}\rho u^2+\varepsilon(\rho,S)\right] +\derv{x}\left[\rho u H\right]=0,
\label{eq:4.12}
\end{equation}
is the Eulerian energy conservation equation.  $H$ in (\ref{eq:4.11}) can be expressed solely 
in terms of the canonical momenta $\pi^t_x=u$, $\pi^m_x=p$ and $S$.  

{\bf In the general de Donder-Weyl theory, the basic idea is} the generalized Legendre 
transformation:
\begin{equation}
H=\pi^t_x x_t+\pi^m_x x_m+\pi^t_S S_t-\tilde{L}(x,x_t,x_m,S,p,r,u,S_t), \label{eq:4.16}
\end{equation}
where the Hamiltonian density $H$ has the form:
\begin{equation}
H=H(x,S,\pi^t_x,\pi^m_x,\pi^t_S). \label{eq:4.16a}
\end{equation}
From (\ref{eq:4.16}) and (\ref{eq:4.16a}) we obtain: 
\begin{align}
dH=&\deriv{H}{x}dx+\deriv{H}{S}dS+\deriv{H}{\pi^t_x}d\pi^t_x+\deriv{H}{\pi^m_x}d\pi^m_x
+\deriv{H}{\pi^t_S}d\pi^t_S\nonumber\\
=&\biggl(x_td\pi^t_x+\pi^t_x dx_t+x_m d\pi^m_x+\pi^m_x dx_m+S_td\pi^t_S+\pi^t_S dS_t\biggr)\nonumber\\
&-\biggl(\deriv{\tilde L}{x}dx+\deriv{\tilde L}{x_t}dx_t+\deriv{\tilde L}{x_m}dx_m+\deriv{\tilde L}{S}dS+\deriv{\tilde L}{S_t}dS_t
+\deriv{\tilde L}{u} du+\deriv{\tilde L}{r}dr+\deriv{\tilde L}{p}dp\biggr). \label{eq:4.17}
\end{align}
From the Euler Lagrange equations (\ref{eq:dw13}) we have:
\begin{align}
&\deriv{\tilde L}{u}=\frac{\delta \tilde{A}}{\delta u}=x_t-u=0, \nonumber\\
&\deriv{\tilde L}{r}=\frac{\delta \tilde{A}}{\delta r}=S_t=0, \nonumber\\
&\deriv{\tilde L}{p}=\frac{\delta \tilde{A}}{\delta p}=x_m-\tau=0. \label{eq:4.17a}
\end{align}
Equating coefficients of $dx_t$, $dx_m$, and $dS_t$ in (\ref{eq:4.17}) gives:
\begin{equation}
\pi^t_x=\deriv{\tilde L}{x_t},\quad \pi^m_x=\deriv{\tilde L}{x_m},\quad \pi^t_S=\deriv{\tilde L}{S_t}, \label{eq:4.18}
\end{equation}
which are the usual definitions for the canonical momenta. Equating the coefficients of 
$d\pi^t_x$, $d\pi^m_x$ and $d\pi^t_S$ in (\ref{eq:4.17}) gives:
\begin{equation}
x_t=\deriv{H}{\pi^t_x},\quad x_m=\deriv{H}{\pi^m_x},\quad S_t=\deriv{H}{\pi^t_S},  \label{eq:4.19}
\end{equation}
which are the canonical evolution equations of $x$ with respect to $t$ and $h$ and the evolution 
equation for $S$ with respect to time. Equating the coefficients of $dx$ and $dS$ equal to zero 
gives the equations:
\begin{equation}
\deriv{\tilde L}{x}+\deriv{H}{x}=0\quad\hbox{and}\quad \deriv{\tilde L}{S}+\deriv{H}{S}=0. \label{eq:4.20}
\end{equation}
In the present example $\partial \tilde{L}/\partial x=\partial H/\partial x=0$  and 
$\partial H/\partial S=\partial\tilde{w}/\partial S=T$ and $\partial \tilde{L}/\partial S=-T$, 
which verifies the validity of (\ref{eq:4.20}). 
The Euler Lagrange equation $\delta{\tilde A}/\delta x=0$ in (\ref{eq:dw13}) can be written in the form:
\begin{equation}
\derv{t}\left(\deriv{\tilde{L}}{x_t}\right)+\derv{m}\left(\deriv{\tilde{L}}{x_m}\right)
-\deriv{\tilde{L}}{x}=0\quad\hbox{or}\quad \deriv{\pi^t_x}{t}+\deriv{\pi^m_x}{h}=-\deriv{H}{x}=0. \label{eq:4.21}
\end{equation}
This equation is `canonically conjugate' to the equations:
\begin{align}
x_t=&\deriv{H}{\pi^t_x}=\deriv{H}{u}=u,\nonumber\\
 x_m=&\deriv{H}{\pi^m_x}=\deriv{H}{p}=\tau. \label{eq:4.22}
\end{align}
The canonically conjugate equations for $\pi^t_S=r$ and $S$ are:
\begin{align}
\deriv{\pi^t_S}{t}=&-\deriv{H}{S}\quad\hbox{or}\quad
 \deriv{r}{t}=-T, \nonumber\\
\deriv{S}{t}=&\deriv{H}{r}=0. \label{eq:4.23}
\end{align}
Thus, the de Donder-Weyl formulation (\ref{eq:4.21})-(\ref{eq:4.23}) are the generalization of Hamilton's equations 
for the multi-momentum case. These equations allow for canonical momenta associated with both the space 
and time (i.e. $m$ and $t$) gradients. Also (\ref{eq:4.21}) for the momentum $\boldsymbol{\pi}_x=(\pi^t_x,\pi^m_x)^T$ 
can be written in the form:
\begin{equation}
\nabla{\bf\cdot}\boldsymbol{\pi}_x=-\deriv{H}{x}\quad\hbox{where}\quad \nabla=(\partial_t,\partial_m)^T, 
\label{eq:4.23a}
\end{equation}
is the gradient operator in $(t,m)$ space.   
\subsection{Multi-symplectic approach}
Following \cite{Hydon05}, we introduce the notation:
\begin{equation}
{\bf z}=(x,u,p,S,r)^T\equiv(x,\pi^t_x,\pi^m_x,S,\pi^t_S)^T, \label{eq:4.24}
\end{equation}
as the dependent variables describing the system. We introduce the one-forms:
\begin{equation}
\omega^\alpha=L_s^\alpha dz^s, \quad\alpha=0,1,\quad s=1,2, \label{eq:4.25}
\end{equation}
where 
\begin{equation}
L_s^\alpha \deriv{z^s}{x^\alpha}=u\deriv{x}{t}+p\deriv{x}{m}+r\deriv{S}{t}, \label{eq:4.26}
\end{equation}
are the terms in the Legendre transform (\ref{eq:4.11}) relating $H$ and $L$ (note the $L^\alpha_s$ 
are the canonical momenta in the de Donder-Weyl approach). 
{\bf In (\ref{eq:4.26}) $x^0=t$ and $x^1=m$.}  
From (\ref{eq:4.25})-(\ref{eq:4.26}) we identify the one forms:
\begin{equation}
\omega^0=udx+rdS\equiv L_1^0 dx+L^0_4 dS, \quad
\omega^1=pdx\equiv L_1^1 dx, \label{eq:4.27}
\end{equation}
as the fundamental one forms describing the system. 

In multi-symplectic systems, the fundamental
2-forms $\kappa^\alpha=d\omega^\alpha$ are given by:
\begin{equation}
\kappa^\alpha=d\omega^\alpha=d(L^\alpha_j dz^j)=\frac{1}{2}{\sf K}^\alpha_{ij} dz^i\wedge dz^j. \label{eq:4.28}
\end{equation}
Thus, 
\begin{equation}
{\sf K}^\alpha_{ij}=\deriv{L^\alpha_j}{z^i}-\deriv{L^\alpha_i}{z^j}, \label{eq:4.29}
\end{equation}
where the matrices ${\sf K}^\alpha_{ij}$ are skew symmetric with respect to the 2 lower indices 
(e.g. \cite{Hydon05}, \cite{Cotter07}, 
\cite{Webb14c}). The condition $\nabla_\alpha\kappa^\alpha=0$,   
 corresponds to the symplecticity or 
conservation of phase space condition for multi-symplectic systems. Here $\nabla_\alpha$ denotes covariant 
differentiation with respect to the independent variables $x^\alpha$ (i.e. $m$ and $t$).  

Taking the exterior derivatives of $\omega^0$ and $\omega^1$ in (\ref{eq:4.27}) gives the 2-forms:
\begin{align}
d\omega^0=&du\wedge dx+dr\wedge dS=dz^2\wedge dz^1+dz^5\wedge dz^4=\frac{1}{2}{\sf K}^0_{ij}dz^i\wedge dz^j, \nonumber\\
d\omega^1=&dp\wedge dx=dz^3\wedge dz^1=\frac{1}{2}{\sf K}^1_{ij}dz^i\wedge dz^j. \label{eq:4.30}
\end{align}
Thus, we identify the non-zero components of the skew-symmetric matrices ${\sf K}^\alpha_{ij}$ as:
\begin{equation}
{\sf K}^0_{21}=1,\quad {\sf K}^0_{12}=-1, \quad {\sf K}^0_{54}=1,\quad {\sf K}^0_{45}=-1, \quad {\sf K}^1_{31}=1,\quad {\sf K}^1_{13}=-1. \label{eq:4.31}
\end{equation}
From the general theory of multi-symplectic systems (\cite{Hydon05}) it follows that (\ref{eq:dw14})
can be written in the multi-symplectic form:
\begin{equation}
{\sf K}^0_{ij}\deriv{z^j}{t}+{\sf K}^1_{ij}\deriv{z^j}{x}=\frac{\delta {\cal H}}{\delta z^i}, \label{eq:4.32}
\end{equation}
where ${\cal H}$ is the Hamiltonian functional:
\begin{equation}
 {\cal H}=\int H dm. \label{eq:4.33}
\end{equation}
In the present analysis 
\begin{equation}
\frac{\delta {\cal H}}{\delta {\bf z}}=\deriv{H}{\bf z}=\left(\deriv{H}{x},\deriv{H}{u},\deriv{H}{p},\deriv{H}{S},\deriv{H}{r}\right)^T 
=\left(0,u,\tau,T,0\right)^T.
\label{eq:4,34}
\end{equation}
The multi-symplectic system (\ref{eq:4.32}) reduces to the matrix system
\begin{equation}
\left(\begin{array}{ccccc}
0&-1&0&0&0\\
1&0&0&0&0\\
0&0&0&0&0\\
0&0&0&0&-1\\
0&0&0&1&0
\end{array}
\right) \derv{t}\left(\begin{array}{c}
x\\
u\\
p\\
S\\
r
\end{array}
\right)
+\left(\begin{array}{ccccc}
0&0&-1&0&0\\
0&0&0&0&0\\
1&0&0&0&0\\
0&0&0&0&0\\
0&0&0&0&0
\end{array}
\right)
\derv{m}
\left(\begin{array}{c}
x\\
u\\
p\\
S\\
r
\end{array}
\right)
=\deriv{H}{\bf z}. \label{eq:4.35}
\end{equation}
The component equations in (\ref{eq:4.35}) are:
\begin{equation}
-(u_t+p_m)=0,\quad x_t=u,\quad x_m=\tau,\quad r_t=-T, \quad S_t=0, \label{eq:4.36}
\end{equation}
which is the system (\ref{eq:dw14}) for 1D, Lagrangian gas dynamics. Note that ${\sf K}^0_{ij}$ 
and ${\sf K}^1_{ij}$ are skew symmetric matrices. 
\subsection{Pullback and symplecticity conservation laws}
From \cite{Hydon05} (see also \cite{Webb14c} and Appendix A), the multi-symplectic system 
(\ref{eq:4.35}) satisfies the pullback conservation laws:
\begin{equation}
D_\alpha\left(L^\alpha_j z^j_{,\beta}-L\delta^\alpha_\beta\right)=0,\quad \alpha,\beta=0,1, \label{eq:4.40}
\end{equation}
where $(x^0,x^1)=(t,m)$ in the present application. The pullback conservation laws can also be derived 
from Noether's theorem and correspond to invariance of the action under translations 
of the independent variables (i.e. the $x^\beta$). 
In (\ref{eq:4.40}) the Lagrangian density $L$ is given by:
\begin{equation}
L=\frac{1}{2}u^2-\varepsilon\tau, \label{eq:4.41}
\end{equation}
where the Lagrangian $L$ is the same as in (\ref{eq:dw2}) except the constraint terms are omitted in (\ref{eq:4.41}). 
The pullback conservation laws are a consequence of the generalized Legendre transformation (e.g. \cite{Hydon05}).

For $\beta=0$, the pullback conservation law (\ref{eq:4.40}) reduces to:
\begin{equation}
D_t I^0+D_m I^1=0, \label{eq:4.42} 
\end{equation}
where
\begin{equation}
I^0=L_j^0z^j_{,0}-L,\quad I^1=L^1_j z^j_{,0}. \label{eq:4.43}
\end{equation}
Using (\ref{eq:4.27}) for $\omega^0$ and $\omega^1$ we obtain:
\begin{align}
I^0=&u\deriv{x}{t}+r\deriv{S}{t}-\left(\frac{1}{2} u^2-\varepsilon\tau\right)\equiv \frac{1}{2} u^2+\varepsilon\tau, 
\nonumber\\
I^1=&p \deriv{x}{t}\equiv p u. \label{eq:4.44}
\end{align}
Thus, we obtain the conservation law:
\begin{equation}
\derv{t}\left(\frac{1}{2} u^2+\varepsilon\tau\right)+\derv{m}\left(pu\right)=0, \label{eq:4.45}
\end{equation}
which is the co-moving energy equation. Using (\ref{eq:3.6}), we obtain:
\begin{equation}
\derv{m}=\deriv{x_0}{m}\derv{x_0}\equiv \frac{1}{\rho_0}\derv{x_0}. \label{eq:4.46}
\end{equation}
Thus, (\ref{eq:4.46}) can also be written in the form:
\begin{equation}
\derv{t}\left[\rho_0\left(\frac{1}{2} u^2+\varepsilon\tau\right)\right] +
\derv{x_0}\left(p u\right)=0. \label{eq:4.47}
\end{equation}

Using 
\begin{equation}
I^0=\rho_0\left(\frac{1}{2} u^2+\varepsilon\tau\right),\quad I^1=pu, \label{eq:4.48}
\end{equation}
for the Lagrangian conserved density $I^0$ and flux $I^1$ in (\ref{eq:4.47}) 
and using 
the transformations for the conserved Eulerian density $F^0$ and flux components $F^j$:
\begin{equation}
F^0=\frac{I^0}{J},\quad F^j=\frac{u^j I^0+x_{jk} I^k}{J}, \label{eq:4.49}
\end{equation}
where $x_{jk}=\partial x^j/\partial x_0^k$ and $J=\det (x_{jk})$ (\cite{Padhye96a, Padhye96b}, \cite{Padhye98}, 
\cite{Webb05}) where $1\leq j,k\leq n$ (note $n=1$ in our case), we obtain the Eulerian energy conservation law:
\begin{equation}
\derv{t}\left[\frac{1}{2}\rho u^2+\varepsilon(\rho,S)\right]
+\derv{x}\left[\rho u\left(\frac{1}{2} u^2+w\right)\right]=0, \label{eq:4.50}
\end{equation}
where $w=(\varepsilon+p)/\rho$ is the gas enthalpy. 

For the case $\beta=1$, the pullback conservation law (\ref{eq:4.40}) reduces to:
\begin{equation}
\deriv{I^0}{t}+\deriv{I^1}{m}=\derv{t}\left(u\tau+rS_m\right)
+\derv{m}\left(w-\frac{1}{2}u^2\right)=0. \label{eq:4.51}
\end{equation}
This conservation law is due to translation invariance of the action with respect
to $m$ in Noether's theorem. Translation in $m$ is a fluid relabelling 
symmetry. This is distinctly different than a translation in the Eulerian position variable  $x$,
which is associated with the conservation of linear momentum in Eulerian coordinates $(t,x)$
as independent variables. 
Using  the transformations (\ref{eq:4.49}) for the case $n=1$ (i.e 1D gas dynamics), we obtain the 
Eulerian, version of the conservation law (\ref{eq:4.51}), namely:
\begin{equation}
\derv{t}\left(u+r\deriv{S}{x}\right)+\derv{x}\left(w+\frac{1}{2}u^2+ur\deriv{S}{x}\right)=0. \label{eq:4.52}
\end{equation}
 Since $dr/dt=-T(x,t)$ then the 
Clebsch variable $r$ is a nonlocal variable, i.e. 
\begin{equation}
r(x,t)=-\int_0^t T(x,t') dt'+r_0({\bf x}_0), \label{eq:4.53}
\end{equation}
is minus the Lagrangian time integral of the temperature  back along the fluid path ($r_0({\bf x}_0)$
is an `integration constant'). 
By noting $dS/dt=(\partial/\partial t+u\partial x)S=0$ and using $dr/dt=-T$ equation (\ref{eq:4.52}) 
gives the Eulerian momentum equation for the fluid as:
\begin{equation}
\deriv{u}{t}+u\deriv{u}{x}=-\frac{1}{\rho}\deriv{p}{x}. \label{eq:4.54}
\end{equation}
Using the mass continuity equation (\ref{eq:2.1}), (\ref{eq:4.54}) can also be written in the momentum 
conservation form:
\begin{equation}
\deriv{(\rho u)}{t}+\derv{x}\left(\rho u^2+p\right)=0. \label{eq:4.55}
\end{equation}
What is perhaps of most interest here, is the existence of a nonlocal conservation law (\ref{eq:4.52})
involving the nonlocal Clebsch potential $r(x,t)$ (see e.g. \cite{Webb14a, Webb14b} for further 
nonlocal conservation laws involving a generalized nonlocal cross helicity conservation equation 
in MHD and a generalized nonlocal helicity conservation equation involving $r({\bf x},t)$ and the entropy 
$S$ of the fluid).  

The $m$-translation conservation law (\ref{eq:4.52}) is connected to the Clebsch variable formulation of fluid mechanics 
(e.g. \cite{Zakharov97}, \cite{Morrison98}, \cite{Webb14c}). In the Clebsch approach the fluid velocity $u$ has the 
form:
\begin{equation}
u=\deriv{\phi}{x}-r\deriv{S}{x}\quad \hbox{where}\quad r=\frac{\beta}{\rho}, \label{eq:4.55a}
\end{equation}
Here $(\beta,S)$ and $(\rho,\phi)$ are canonically conjugate variables and $\phi$ is the velocity potential 
corresponding to potential flow  (e.g. \cite{Zakharov97}). The potential $\phi$ satisfies Bernoulli's 
equation:
\begin{equation}
\deriv{\phi}{t}+u \deriv{\phi}{x}=\frac{1}{2} u^2-w. \label{eq:4.55b}
\end{equation}
(\ref{eq:4.52}) has the conservation equation form:
\begin{equation}
D_t+F_x=0, \label{eq:4.55c}
\end{equation}
in which:
\begin{equation}
D=u+r\deriv{S}{x}=\deriv{\phi}{x}, \quad F=w+\frac{1}{2}u^2+ur\deriv{S}{x}=-\deriv{\phi}{t}. \label{eq:4.55d}
\end{equation} 
Thus, the conservation law (\ref{eq:4.52}), written in terms of the velocity potential $\phi$ reduces to the 
equation $\phi_{xt}-\phi_{tx}=0$, i.e. (\ref{eq:4.52}) can be written as the integrability condition 
$\phi_{xt}-\phi_{tx}=0$ for $\phi(x,t)$.  
Similarly, the Lagrangian form of the conservation law (\ref{eq:4.51}) 
can be written in the form: $\phi_{\theta m}-\phi_{m \theta}=0$ where 
$\partial/\partial \theta=\partial/\partial t+u\partial/\partial x$ is the Lagrangian time 
derivative following the flow. 
Note that
\begin{equation}
\phi_m=\tau (u+r S_x)\quad\hbox{and}\quad \phi_\theta=\frac{1}{2} u^2-w. \label{eq:4.55e}
\end{equation}
 
\subsubsection{Symplecticity conservation law}
The symplecticity conservation law for the multi-symplectic system (\ref{eq:4.32}) is the set of 
conservation laws:
\begin{equation}
D_\alpha\left(F^\alpha_{\beta\gamma}\right)=0, \quad \beta<\gamma, \label{eq:4.56}
\end{equation}
where
\begin{equation}
F^\alpha_{\beta\gamma}={\bf z}^T_{,\beta}{\sf K}^\alpha {\bf z}_{,\gamma}
\equiv {\sf K}^\alpha_{ij}z^i_{,\beta}z^j_{,\gamma}. \label{eq:4.57}
\end{equation}
These conservation laws can be derived from the pullback (i.e. treat the dependent variables 
as functions of the independent variables) of the symplecticity conservation condition 
$\kappa^\alpha_{,\alpha}=0$, where
\begin{equation}
\kappa^\alpha=d\omega^\alpha=d\left(L^\alpha_j dz^j\right)=\frac{1}{2}{\sf K}^\alpha_{ij}dz^i\wedge dz^j, 
\label{eq:4.58}
\end{equation}
(Hydon (2005), see also Appendix A). The symplecticity conservation laws (\ref{eq:4.56}) are related to 
the pullback conservation laws (\ref{eq:4.40}) by the equations:
\begin{equation}
D_\gamma G_\beta-D_\beta G_\gamma=D_\alpha\left(F^\alpha_{\gamma\beta}\right), \quad\gamma<\beta, \label{eq:4.59}
\end{equation}
where
\begin{equation}
G_\beta=D_\alpha\left(L^\alpha_j z^j_{,\beta}-L\delta ^\alpha_\beta\right). \label{eq:4.60}
\end{equation}
Note that $G_\beta=0$ are the pullback conservation laws (\ref{eq:4.40}). 

From (\ref{eq:4.56})-(\ref{eq:4.60}), there is only one symplecticity conservation law in the present case. 
The pullback symplecticity conservation law (\ref{eq:4.56}) reduces 
to the equation:
\begin{equation}
D_\alpha\left(F^\alpha_{01}\right)=D_t\left(u_tx_m-x_t u_m+r_t S_m\right) +D_m\left(p_t x_m-x_t p_m\right)=0. 
\label{eq:4.61}
\end{equation}
{\bf The symplecticity conservation law (\ref{eq:4.61}) can be written in the form:
\begin{equation}
\deriv{D}{t}+\deriv{F}{m}=\derv{t}\left[\frac{\partial(u,x)}{\partial(t,m)}
+\frac{\partial(r,S)}{\partial(t,m)}\right]
+\derv{m}\left(\frac{\partial(p,x)}{\partial(t,m)}\right)=0, \label{eq:4.63a}
\end{equation}
where $\partial(\phi,\psi)/\partial(\alpha,\beta)$ is the Jacobian of $\phi$ and $\psi$ with respect to 
$\alpha$ and $\beta$. One can show that (\ref{eq:4.63a}) can be written in the form:
\begin{equation}
\deriv{D}{t}+\deriv{F}{m}= -\derv{m}\left[\derv{t}\left(\frac{1}{2} u^2+e\right)+\derv{m}\left(pu\right)\right]=0, \label{eq:4.63b}
\end{equation}
where $e=\varepsilon\tau$  is the internal energy density of the gas per unit mass and $\tau=1/\rho$. 
Equation (\ref{eq:4.63b}) is minus the partial derivative with respect to $m$ 
of the co-moving energy equation (\ref{eq:4.45}). 
} 
The conserved density $D$ and flux $F$ in (\ref{eq:4.63a}) can also be reduced to:
\begin{equation}
D=-H_m,\quad F=H_t, \label{eq:4.63c}
\end{equation}
where $H=(1/2)u^2+w$ is the multi-symplectic Hamiltonian. Thus, 
the symplecticity conservation law (\ref{eq:4.63a}) 
is also equivalent to the equation $-H_{mt}+H_{tm}=0$. 
\subsection{Noether's theorem}
A multi-symplectic form of Noether's theorem was described by \cite{Webb14c}. The basic idea is 
that if the action:
\begin{equation}
A=\int\int L\ dm dt, \label{eq:4.64}
\end{equation}
is invariant under the infinitesimal Lie transformations:
\begin{equation}
m'=m+\epsilon V^m,\quad t'=t+\epsilon V^t,\quad  z^{'s}=z^s+\epsilon V^{z^s}, \label{eq:4.65}
\end{equation}
and under the divergence transformation:
\begin{equation}
L'=L+\epsilon D_\alpha \Lambda^\alpha, \label{eq:4.66}
\end{equation}
then Noether's theorem for the multi-symplectic system implies the conservation law:
\begin{equation}
D_t\left(V^t L+\hat{V}^{z^s} L^0_s+\Lambda^0\right)+D_m\left(V^m L+\hat{V}^{z^s} L^1_s+\Lambda^1\right)=0, 
\label{eq:4.67}
\end{equation}
where
\begin{equation}
\hat{V}^{z^s}=V^{z^s}-\left(V^tD_t+V^mD_m\right)z^s, \label{eq:4.68}
\end{equation}
is the canonical or characteristic symmetry generator for transformations of $z^s$ of the form:
$z^{'s}=z^s+\epsilon \hat{V}^{z^s}$,\ $t'=t$,\ $m'=m$ that are equivalent to the transformations 
(\ref{eq:4.65})-(\ref{eq:4.66}). 

\begin{example}
Consider the infinitesimal Lie transformations (\ref{eq:4.65})-(\ref{eq:4.66}) with:
\begin{equation}
V^t=0,\quad V^m=1,\quad V^x=0,\quad \Lambda^\alpha=0,\quad (\alpha=0,1), \label{eq:4.69}
\end{equation} 
Here ${\bf z}=(x,u,p,S,r)^T$ (see (\ref{eq:4.24})). The action is invariant under the transformations
with generators (\ref{eq:4.69}), and with the $L^\alpha_j$ given in (\ref{eq:4.27}) 
where $\omega^\alpha=L^\alpha_j dz^j$ and $L^0_1=u$, $L^0_4=r$, $L^1_1=p$ are the non-zero $L^\alpha_j$. 
Using (\ref{eq:4.68}) and (\ref{eq:4.69}) we obtain: 
$\hat{V}^x=-x_m$, $\hat{V}^S=-S_m$. The Noether conservation law (\ref{eq:4.67})
gives:
\begin{equation}
D_t\left(\hat{V}^x u+\hat{V}^S r\right)+D_m\left(L+\hat{V}^x p\right)=0, \label{eq:4.70}
\end{equation}
which reduces to the nonlocal conservation law (\ref{eq:4.51}). In other words, (\ref{eq:4.51}) follows 
from the $m$-translation invariance of the action (\ref{eq:4.64}) and Noether's theorem. 

Similarly, the Lie transformations (\ref{eq:4.65})-(\ref{eq:4.66}) with:
\begin{equation}
V^t=1,\quad V^m=0,\quad V^x=0,\quad \Lambda^0=\Lambda^1=0, \label{eq:4.71}
\end{equation}
leaves the action (\ref{eq:4.64}) invariant under time translations and give rise to the energy 
conservation law (\ref{eq:4.45}). The transformations (\ref{eq:4.65})-(\ref{eq:4.66}) with:
\begin{equation}
V^t=0,\quad V^m=0, \quad V^x=1,\quad \Lambda^0=\Lambda^1=0, \label{eq:4.72}
\end{equation}
leave the action (\ref{eq:4.64}) invariant, corresponding to the $x$-translation invariance symmetry 
and give rise, via Noether's theorem to the $x$-momentum conservation law:
\begin{equation}
u_t+p_m=0, \label{eq:4.73}
\end{equation}
which is the Lagrangian form of the $x$-momentum conservation equation. 
\end{example}
\subsection{Differential forms formulation}
\begin{proposition}
The condition that the action:
\begin{equation}
J=\int \psi^*(\Theta)\equiv \int \tilde{L} dV,  \label{eq:4.75}
\end{equation}
is stationary, namely $\delta J/\delta z^i=0$
where $\Theta$ is the 2-form defined by the equations:
\begin{align}
\Theta=&\omega^\alpha\wedge d\tilde{m}_\alpha-H dV, \label{eq:4.75aa}\\
dV=&dt\wedge dm, \quad d\tilde{m}_\alpha=\partial_\alpha\lrcorner dV, \quad d\tilde{m}_0=dm, \quad
d\tilde{m}_1=-dt, \label{eq:4.75ba}
\end{align}
( $m^0=t$ and $m^1=m$)
 gives the multi-symplectic system (\ref{eq:4.32})-(\ref{eq:4.35}). In (\ref{eq:4.75}) 
 $\tilde{L}$ is the multi-symplectic Lagrangian (\ref{eq:dw9}). 
We obtain:
\begin{align}
\Theta=&\omega^0\wedge dm-\omega^1\wedge dt-H dt\wedge dm\nonumber\\
\equiv&\pi_x^t dx\wedge d\tilde{m}_0+\pi^t_S dS\wedge d\tilde{m}_0+ \pi_x^m dx\wedge d\tilde{m}_1
-H dt\wedge dm,\nonumber\\
\equiv&(udx+rdS)\wedge dm-pdx\wedge dt-HdV.
\end{align}
Hence the pullback of $\Theta$ is given by:
\begin{equation}
\psi^*(\Theta)=(u x_t+r S_t+p x_m-H)dV\equiv 
\left(L^\alpha_s 
\deriv{z^s}{m^\alpha}-H\right)dV=\tilde{L} dV.
 \label{eq:4.75a}
\end{equation}
\end{proposition}
\begin{proof}
{\bf  The proof that $\delta J/\delta z^i=0$ gives the multi-symplectic system (\ref{eq:4.35}) 
follows by noting  $\psi^*(\Theta)=\tilde{L} dV$ and using the analysis of Hydon (2005) to obtain:
\begin{equation}
\frac{\delta J}{\delta z^i}=\deriv{\tilde{L}}{z^i}
-\derv{m^j}\left(\deriv{\tilde L}{(\partial_j z^i)}\right)
={\sf K}^\alpha_{ij}\deriv{z^j}{m^\alpha}-\deriv{H}{z^i}=0. 
\label{eq:4.77}
\end{equation}
}
\end{proof}
\begin{proposition}
Consider the variational functional 
\begin{equation}
K[\Omega]=\int_M \Omega, \label{eq:vbun1}
\end{equation}
where $M$ is a region with boundary $\partial M$ in the fiber bundle space in which the $z^i$ are 
regarded as independent of the base manifold coordinates $m^\alpha=(t,m)$. The form:
\begin{equation}
\Omega=d\Theta=d\omega^\alpha\wedge d\tilde{m}_\alpha-dH\wedge dV
\quad\hbox{where}\quad dV=dt\wedge dm, \label{eq:vbun2}
\end{equation}
is known as the Cartan-Poincar\'e form. Consider variations of the functional (\ref{eq:vbun1}) 
described by the Lie derivative:
\begin{equation}
{\cal L}_{\bf V}=\frac{d}{d\epsilon}=V^i\derv{z^i}, \label{eq:vbun3}
\end{equation}
in which the base manifold variables $m^\alpha$ are fixed. The vector field ${\bf V}$ 
is an arbitrary but smooth vector field. The variations of $K[\Omega]$ 
described by:
\begin{equation}
\delta K[\Omega]=\int_M {\cal L}_{\bf V}\left(\Omega\right), \label{eq:vbun4}
\end{equation}
can be reduced to the form:
\begin{equation}
\delta K[\Omega]=\int_{\partial M} V^p\beta_p, \label{eq:vbun5}
\end{equation}
where the forms $\{\beta_p:\ 1\leq p\leq N\}$ ($N$ is the number of $z^i$ variables) are given 
by:
\begin{equation}
\beta_p={\sf K}^\alpha_{pj} dz^j\wedge d\tilde{m}_\alpha-\deriv{H}{z^p} dV, \label{eq:vbun6}
\end{equation}
and $\partial M$ is the boundary of the region $M$ in the ${\bf z}$-space. 
The set of equations $\beta_p=0$ ($1\leq p\leq N)$, can be used as a basis of Cartan 
forms for the multi-symplectic system (\ref{eq:4.32})-(\ref{eq:4.35}). The pullback of the
differential forms 
 $\beta_p=0$ to the base manifold  gives the equations:
\begin{equation}
\tilde{\beta}_p=\left({\sf K}^\alpha_{pj}\deriv{z^j}{m^\alpha}-\deriv{H}{z^p}\right) dV. 
\label{eq:vbun7}
\end{equation}
The sectioned forms $\tilde{\beta}_p$ vanish on the solution manifold of the multi-symplectic 
partial differential equation system (\ref{eq:4.32})-(\ref{eq:4.35}).  
\end{proposition}
\begin{proof}
{\bf To prove (\ref{eq:vbun5}) from (\ref{eq:vbun4}) we use Cartan's magic formula:
\begin{equation}
{\cal L}_{\bf V}\Omega={\bf V}\lrcorner d\Omega+d\left({\bf V}\lrcorner \Omega\right)
\equiv d\left({\bf V}\lrcorner \Omega\right), \label{eq:vbun8}
\end{equation}
where we have used the fact that $\Omega=d\Theta$ and $d\Omega=dd\Theta=0$. Using (\ref{eq:vbun8}) 
in (\ref{eq:vbun4}) and using Stokes' theorem, we obtain:
\begin{equation}
\delta K[\Omega]=\int_M d\left({\bf V}\lrcorner \Omega\right)
=\int_{\partial M} \left({\bf V}\lrcorner \Omega\right). \label{eq:vbun9}
\end{equation}
Using (\ref{eq:vbun2}) for $\Omega$ and noting that
\begin{equation}
d\omega^\alpha=\frac{1}{2} {\sf K}^{\alpha}_{ij} dz^i\wedge dz^j,
\quad dH=\deriv{H}{z^i} dz^i\wedge dV, 
\label{eq:vbun10}
\end{equation}
we obtain:
\begin{equation}
\Omega=\frac{1}{2} {\sf K}^{\alpha}_{ij} dz^i\wedge dz^j\wedge d\tilde{m}_\alpha
-\deriv{H}{z^i} dz^i\wedge dV. 
\label{eq:vbun11}
\end{equation}
Thus,
\begin{align}
{\bf V}\lrcorner\Omega=&V^p\derv{z^p}\lrcorner\left(\frac{1}{2} {\sf K}^\alpha_{ij} dz^i\wedge dz^j
\wedge d\tilde{m}_\alpha
-\deriv{H}{z^i} dz^i\wedge dV\right)\nonumber\\
=&V^p\left({\sf K}^\alpha_{pj} dz^j\wedge d\tilde{m}_\alpha
-\deriv{H}{z^p} dV\right)=V^p\beta_p. \label{eq:vbun12}
\end{align}
In deriving (\ref{eq:vbun12}) we used the anti-symmetry of the matrices ${\sf K}^\alpha_{ij}$ 
and the anti-symmetry of the wedge product to obtain the final result. 
Substituting (\ref{eq:vbun12}) in (\ref{eq:vbun9})
now gives (\ref{eq:vbun5}) and (\ref{eq:vbun6}) for $\delta K[\Omega]$. 
Note that the variational principle 
\begin{equation}
\delta K[\Omega]=\int_{\partial M} V^p\beta_p=0\quad \hbox{implies}\quad \beta_p=0, \label{eq:vbun13}
\end{equation}
for $1\leq p\leq N$, since the $V^p$ are arbitrary smooth functions. 
The closure of the differential form system of equations $\beta_p=0$ in 
(\ref{eq:vbun13}) is a set of Cartan forms representing the original 
multi-symplectic partial differential system (\ref{eq:4.36}) (see e.g. \cite{Harrison71}). 

From (\ref{eq:vbun6}) we obtain the pullback forms:
\begin{equation}
\tilde{\beta}_p={\sf K}^\alpha_{pj} \deriv{z^j}{m^s} dm^s\wedge d\tilde{m}_\alpha-\deriv{H}{z^p} dV. 
\label{eq:vbun14}
\end{equation}
However, if there are $n$ spatial $m^\alpha$ components, then
\begin{align}
dm^s\wedge d\tilde{m}_\alpha=&dm^s\wedge (-1)^\alpha dm^0\wedge\ldots
\wedge dm^{\alpha-1}\wedge dm^{\alpha+1}\wedge\ldots dm^n\nonumber\\
=&dm^s \delta^s_\alpha (-1)^{2\alpha} dV=\delta^s_\alpha dV. \label{eq:vbun15}
\end{align}
For 1D gas dynamics, $n=1$. 
Using (\ref{eq:vbun15}), the pullback forms (\ref{eq:vbun14}) reduce to the forms (\ref{eq:vbun7}). 
This completes the proof.
}
\end{proof}

{\bf  In the case of 1D Lagrangian, compressible gas dynamics, the forms (\ref{eq:vbun6}) reduce to:
\begin{equation}
\beta_p={\sf K}^0_{pj} dz^j\wedge dm-{\sf K}^1_{ij} dz^j\wedge dt
-\deriv{H}{z^i} dt\wedge dm, \quad 1\leq p\leq 5. \label{eq:vbun16}
\end{equation}
The closure of the system of forms
$\{\beta_i:\quad 1\leq i\leq 5\}$,
comprise a Cartan differential forms representation of multi-symplectic 
differential equation  system (\ref{eq:4.36}). 
Sectioning the forms (i.e. applying the pullback to the forms) results in the system (\ref{eq:4.36}),
i.e. $\tilde{\beta}_i=0$ gives (\ref{eq:4.36}), where the tilde means sectioning the forms.
}

The two-forms $\{\beta_1,\beta_2,\beta_3,\beta_4,\beta_5\}$ in (\ref{eq:vbun16}) for the 1D gas 
dynamic system are: 
\begin{align}
\beta_1=&-du\wedge dm+dp\wedge dt,\quad \beta_2=(dx-u dt)\wedge dm,\nonumber\\
\beta_3=&dt\wedge(dx-\tau dm),\quad \beta_4=-(dr+T dt)\wedge dm,\quad \beta_5=dS\wedge dm. 
\label{eq:4.82}
\end{align}
In (\ref{eq:4.82}) the basic variables are ${\bf z}=(x,u,p,S,r)^T$. The specific volume $\tau$ 
and the temperature 
of the gas $T$ are functions of the enthalpy $\tilde{w}(p,S)\equiv w(\rho,S)$, and are given by the formulae:
\begin{equation}
\tilde{w}_p=\tau,\quad \hbox{and}\quad \tilde{w}_S=T. \label{eq:4.83}
\end{equation}

We now show, that the 2-forms (\ref{eq:4.82}) form a closed ideal of forms that represent the gas dynamics system 
(\ref{eq:4.36}) (e.g. \cite{Harrison71}). This means that the exterior derivatives of the basis forms 
$\{\beta_i:\ 1\leq i\leq 5\}$ can be written as a combination of the basis forms (\ref{eq:4.82}) in the form:
\begin{equation}
d\beta_i=c_{ij}\wedge \beta_j, \label{eq:4.84}
\end{equation}
where the $c_{ij}$ are 1-forms (some of the $c_{ij}$ are in fact zero). Straightforward evaluation of 
$d\beta_1$, $d\beta_2$ and $d\beta_5$ give the equations:
\begin{equation}
d\beta_1=0,\quad d\beta_2=\beta_1\wedge dt,\quad d\beta_5=0. \label{eq:4.85}
\end{equation}
Using (\ref{eq:4.84}) gives:
\begin{align}
d\tau=&d\tilde{w}_p=\tilde{w}_{pp}dp+\tilde{w}_{pS}dS,\nonumber\\
dT=&d\tilde{w}_S=\tilde{w}_{Sp}dp+\tilde{w}_{SS}dS, \label{eq:4.86}
\end{align}
and using these results in (\ref{eq:4.82}) gives:
\begin{align}
d\beta_3=&-\left(\tilde{w}_{pp}dp+\tilde{w}_{pS}dS\right)\wedge dt\wedge dm,\nonumber\\
=& -\tilde{w}_{pp}\beta_1\wedge dm+\tilde{w}_{pS}\beta_5\wedge dt. \label{eq:4.87}
\end{align}
A similar calculation gives:
\begin{equation}
d\beta_4=\tilde{w}_{SS} \beta_5\wedge dt-\tilde{w}_{Sp}\beta_1\wedge dm. \label{eq:4.88}
\end{equation}
Equations (\ref{eq:4.85}), (\ref{eq:4.87}) and (\ref{eq:4.88}) 
demonstrates that the exterior derivatives of 
the forms $I=\left\{\beta_1,\beta_2,\beta_3,\beta_4,\beta_5\right\}$ 
lie in the ideal $I$, i.e. the set $I$ is a closed ideal of 
forms that represent the 1D gas dynamic equations (\ref{eq:4.36}). We can in fact, replace the 2-forms 
$\beta_1$ and $\beta_5$ in the ideal $I$ by 1-forms $\alpha_1$ and $\alpha_5$ by noting:
\begin{align}
\beta_1=&d\alpha_1\quad\hbox{where}\quad \alpha_1=-udm+p dt, \nonumber\\
\beta_5=&d\alpha_5\quad\hbox{where}\quad \alpha_5=Sdm. \label{eq:4.89}
\end{align}

At this point we could proceed to obtain the Lie symmetries and conservation laws of (\ref{eq:4.36}) 
using the ideal $I$ and the methods of \cite{Harrison71} and \cite{Wahlquist75}.
However, this lies beyond the scope of the present paper, and will be investigated in a separate paper. 
Here, we merely note that for the case of an ideal gas with constant adiabatic index $\gamma$ with equation of state:
\begin{equation}
p=p_0\left(\frac{\tau}{\tau_0}\right)^{-\gamma}\exp\left(\frac{S}{C_v}\right), \label{eq:4.90}
\end{equation}
the enthalpy $\tilde{w}(p,S)$ is given by:
\begin{equation}
\tilde{w}(p,S)=\frac{a_0^2}{\gamma-1} \left(\frac{p}{p_0}\right)^{(\gamma-1)/\gamma}
\exp\left(\frac{S}{\gamma C_v}\right), \label{eq:4.91}
\end{equation}
where $a_0^2=\gamma p_0/\rho_0$ is the square of the sound speed of the gas in $0$-state.

\section{Summary and Concluding Remarks}
\setcounter{equation}{0}
In this paper we cast the equations of ideal, Lagrangian 1D fluid dynamics in a multi-symplectic form 
(equations (\ref{eq:4.33})-(\ref{eq:4.35})). We used $m$ (the Lagrangian mass coordinate) 
and time $t$ as the independent variables. The dependent variables used were ${\bf z}=(x,u,p,S,r)^T$
where $x$ is the Eulerian particle position, $u=x_t$ is the fluid velocity, $p$ is the gas pressure, $S$ is the 
entropy and $r$ is the Clebsch variable or Lagrange multiplier 
in the modified Lagrangian which ensures $\partial S(m,t)/\partial t=0$. The 
equations are related to a multi-momentum, de Donder-Weyl Hamiltonian formulation (Section 4.1) in 
which $u=x_t=\partial L/\partial x_t$, $p=\partial L/\partial x_m$ and $r=\partial L/\partial S_t$
are the momentum variables and $L$ is the Lagrangian. The same equations are also cast in a multi-symplectic
form (\ref{eq:4.33})-(\ref{eq:4.35}), in which the Hamiltonian $H=(1/2)u^2+w\equiv F_E/(\rho u)$ is the  
ratio of the Eulerian energy flux ($F_E$) to the mass flux ($\rho u$) and $w=(\varepsilon+p)/\rho$ is the 
gas enthalpy.

We obtained the pullback conservation laws of the multisymplectic, Lagrangian gas dynamics system 
and the symplecticity conservation law (see e.g. \cite{Hydon05}). The 
pullback conservation laws can also be obtained from Noether's theorem for the multisymplectic
system, and correspond to the invariance of the action under time translation (giving rise to the energy
conservation law), and a conservation law associated with translations in $m$. The conservation law
due to translations in $m$ is a fluid relabelling symmetry 
conservation law. It is 
 a nonlocal conservation law,  involving the nonlocal Clebsch variable $r$ which is the
Lagrange multiplier used to ensure that $S$ is advected with the flow, (i.e. $S_t=0$).
The variable $r$ satisfies the evolution equation $r_t=-T$ where $T$ is the temperature of the gas. 
The $m$-translation conservation law (\ref{eq:4.51}) or its Eulerian 
version (\ref{eq:4.52}) takes its most simple form by using the 
Clebsch representation for the fluid velocity,
namely $u=\nabla\phi-r\nabla S$.  Conservation law (\ref{eq:4.52}) reduces 
to the equation $\phi_{xt}-\phi_{tx}=0$. Thus, both 
the Eulerian Clebsch representation (e.g. \cite{Cotter07}, 
\cite{Webb14c}) and the Lagrangian versions of the multi-symplectic fluid equations (this paper)
are useful. 
{\bf The symplecticity law  (\ref{eq:4.63a}) is equivalent to  a compatibility condition 
on the pullback conservation laws (see (\ref{eq:4.59})). The symplecticity 
 law (\ref{eq:4.63a}) is equivalent to minus $\partial/\partial m$ operating 
on the co-moving energy equation (\ref{eq:4.45}).} It is also 
equivalent to the the equation $H_{mt}-H_{tm}=0$ 
where $H=(1/2) u^2 +w$ is the multi-symplectic Hamiltonian using the Lagrangian variables.

The Cartan-Poincar\'e $(n+2)$-form $\Omega$ ($n=1$ in 1D gas dynamics) 
gives rise via the variational principle (\ref{eq:vbun1})-(\ref{eq:vbun4}), to a set 
of two-forms $I=\{\beta_1,\beta_2,\beta_3,\beta_4,\beta_5\}$ representing the 1D gas dynamics 
equations.  
 The set $I$ is a closed ideal of forms, which  may be used to 
determine the Lie symmetries 
and conservation laws of the system using Cartan's geometric theory of partial differential 
equations (e.g. \cite{Harrison71}). The Lie point symmetries and potential symmetries
of the gas dynamic equations has been investigated by \cite{Akhatov91}, \cite{Cheviakov08}, 
\cite{Bluman06, Bluman10}, and \cite{Sjoberg04}. The potential symmetries can lead to 
nonlocal conservation laws of the system. 
\cite{Webb09} investigated conservation laws associated 
with the scaling symmetries of the 1D gas dynamic equations 
and with the nonlocal symmetry of the potential 
cover system of equations obtained by \cite{Sjoberg04}. 

\cite{Webb14c} studied the 1D multi-symplectic gas dynamic equations, 
by using Clebsch variables. 
In this paper, the multi-symplectic, Lagrangian gas dynamic equations were obtained.
 In general, we expect that there is a map between the Lagrangian and Eulerian 
 multi-symplectic gas dynamic equations. This subject lies beyond the scope of the present paper. 
It is of interest to extend the present 1D analysis to the fully 3D compressible 
fluid dynamics and magnetohydrodynamics 
(MHD) equations, using for example the work of \cite{Newcomb62} 
who obtained variational principles for the MHD
 using the Lagrangian map, without using Clebsch variables. 
A recent comprehensive overview of multi-symplectic 
systems by \cite{Roman-Roy09} contains other points of the theory that have not been addressed in the 
present development. For example, \cite{Kanatchikov93, Kanatchikov97, Kanatchikov98} and \cite{Forger05} discuss the covariant 
Poisson brackets
for multi-symplectic and de Donder-Weyl Hamiltonian systems (see also \cite{Marsden86} 
who write the field 
equations in a single Poisson bracket form $\{F,S\}=0$ where $S$ is the action integral with 
applications to electrodynamics, 
relativistic Maxwell-Vlasov equations in plasma physics, and in general relativity).

\section*{Acknowledgements}
I would like to thank the referee for a thorough report, that pointed out inconsistencies 
in the propositions IV.I and IV.2 in the original manuscript.
I acknowledge discussions with Darryl Holm on fluid relabelling symmetries, 
Lagrangian fluid mechanics and 
multi-symplectic formulations of fluid mechanics. I acknowledge discussions with Phil. Morrison 
on multi-symplectic MHD. I thank Prof. QuanMing Lu (CAS, Key Lab. for Geospace environment, USTC, in Hefei, China)
 for financial support to attend the GAMP Geometric Algorithms 
and Methods in Plasma Physics Conference, Hefei, Anhui Province, May 13-15, 2014, where I presented 
 work on Eulerian, multi-symplectic MHD using Clebsch variables. I acknowledge discussions with 
J. F. McKenzie and G.P. Zank. 
\appendix
\section*{Appendix A}
\renewcommand{\theequation}{A.\arabic{equation}}
\chapter{}
In this appendix we indicate the derivation of the pullback conservation laws (\ref{eq:4.40}) 
and the symplecticity conservation laws (\ref{eq:4.56})-(\ref{eq:4.57}) (see \cite{Hydon05} for more detail). 

For multi-symplectic systems with associated 1-forms $\omega^\alpha=L^\alpha_j dz^j$, the 
equations are:
\begin{equation}
{\sf K}^\alpha_{ij}\deriv{z^j}{x^\alpha}=\deriv{H({\bf z})}{z^i},\quad {\sf K}^\alpha_{ij}=\deriv{L^\alpha_j}{z^i}
-\deriv{L^\alpha_i}{z^j}. \label{eq:A1}
\end{equation}
The generalized Legendre transformation for this system is:
\begin{equation}
\left(L^\alpha_j dz^j\right)_{,\alpha}=dL, \label{eq:A2}
\end{equation} 
where 
\begin{equation}
L=L^\alpha_j z^j_{,\alpha}-H({\bf z}). \label{eq:A3}
\end{equation}
The pullback of (\ref{eq:A2}) in which $dz^j=(\partial z^j/\partial x^\beta) dx^\beta$ 
and $dL=(\partial L/\partial x^\beta)dx^\beta$ gives the pullback conservation law (\ref{eq:4.40}). 
Similarly, the symplecticity condition:
\begin{equation}
\kappa^\alpha_{,\alpha}=0\quad\hbox{where}\quad 
\kappa^\alpha=d\omega^\alpha=\frac{1}{2}{\sf K}^\alpha_{ij}dz^i\wedge dz^j, \label{eq:A4}
\end{equation}
pulled back to the base manifold with independent variables $x^\beta$, gives:
\begin{equation}
\kappa^\alpha_{,\alpha}=\frac{1}{2}\left({\sf K}^\alpha_{ij}\deriv{z^i}{x^\beta}\deriv{z^j}{x^\gamma}dx^\beta\wedge dx^\gamma\right)_{,\alpha}=0, \label{eq:A5}
\end{equation}
which implies the symplecticity conservation laws (\ref{eq:4.56})-(\ref{eq:4.57}).

\end{document}